\documentclass[a4paper,UKenglish]{lipics}

\usepackage{microtype}

\usepackage{amssymb}

\usepackage{algorithm}
\usepackage[noend]{algorithmic}
\usepackage{xpatch}
\xpatchcmd{\algorithmic}{\setcounter}{\algorithmicfont\setcounter}{}{}
\providecommand{\algorithmicfont}{}
\providecommand{\setalgorithmicfont}[1]{\renewcommand{\algorithmicfont}{#1}}

\usepackage{wrapfig}
\usepackage{xspace}

\newtheorem{fact}[theorem]{Fact}
\newcommand{\MYCOMMENT}[1]{}
\newcommand{\mylabel}[1]{\label{#1}}
\theoremstyle{plain}

\newcommand{\mycompact}[2]{\ensuremath{\langle #1(i) \rangle_{i\in[1,#2]}}}

\newcommand{\history}[1]{\ensuremath{\mathbf{#1}}\xspace}
\newcommand{\thrhist}[1]{\ensuremath{\vec{\mathbf{#1}}}}
\newcommand{\precOrder}[1]{\ensuremath{\prec_{#1}}}
\newcommand{\abOrder}[1]{\ensuremath{<_{#1}}}

\newcommand{\Evt}[1]{\ensuremath{E_{#1}}}

\newcommand{\LTS}{\ensuremath{\mathsf{LTS}}}
\newcommand{\LTSx}[1]{\ensuremath{\LTS_{#1}}}
\newcommand{\mytrace}{\ensuremath{tr}}
\newcommand{\mytracex}[1]{\ensuremath{\mytrace(#1)}}
\newcommand{\myTr}{\ensuremath{\mathsf{Tr}}}
\newcommand{\myTrx}[1]{\ensuremath{\myTr(#1)}}

\newcommand{\NULL}{\ensuremath{\mathtt{NULL}}}

\newcommand{\mypool}{\ensuremath{\mathcal{P}}}
\newcommand{\mypoolmem}{\ensuremath{\mathcal{P^?}}}
\newcommand{\myqueue}{\ensuremath{\mathcal{Q}}}
\newcommand{\mystack}{\ensuremath{\mathcal{S}}}
\newcommand{\myreg}{\ensuremath{\mathcal{R}}}

\newcommand{\myput}{\ensuremath{\mathtt{put}}\xspace}
\newcommand{\myputx}[1]{\ensuremath{\myput(#1)}\xspace}
\newcommand{\mytake}{\ensuremath{\mathtt{take}}\xspace}
\newcommand{\mytakex}[1]{\ensuremath{\mytake(#1)}\xspace}
\newcommand{\myPut}{\ensuremath{\mathtt{Put}}}
\newcommand{\myTake}{\ensuremath{\mathtt{Take}}}

\newcommand{\mymem}{\ensuremath{\mathtt{mem}}\xspace}
\newcommand{\mymemx}[2]{\ensuremath{\mymem(#1,#2)}\xspace}
\newcommand{\myMem}{\ensuremath{\mathtt{Mem}}}

\newcommand{\mydeq}{\ensuremath{\mathtt{deq}}}
\newcommand{\mydeqx}[1]{\ensuremath{\mydeq(#1)}}
\newcommand{\myenq}{\ensuremath{\mathtt{enq}}}
\newcommand{\myenqx}[1]{\ensuremath{\myenq(#1)}}
\newcommand{\myEnq}{\ensuremath{\mathtt{Enq}}}
\newcommand{\myDeq}{\ensuremath{\mathtt{Deq}}}

\newcommand{\mypush}{\ensuremath{\mathtt{push}}}
\newcommand{\mypushx}[1]{\ensuremath{\mypush(#1)}}
\newcommand{\mypop}{\ensuremath{\mathtt{pop}}}
\newcommand{\mypopx}[1]{\ensuremath{\mypop(#1)}}
\newcommand{\myPush}{\ensuremath{\mathtt{Push}}}
\newcommand{\myPop}{\ensuremath{\mathtt{Pop}}}

\newcommand{\mywr}{\ensuremath{\mathtt{wr}}}
\newcommand{\mywrx}[2]{\ensuremath{\mywr(#1,#2)}}
\newcommand{\myrd}{\ensuremath{\mathtt{rd}}}
\newcommand{\myrdx}[2]{\ensuremath{\myrd(#1,#2)}}
\newcommand{\myWr}{\ensuremath{\mathtt{Wr}}}
\newcommand{\myRd}{\ensuremath{\mathtt{Rd}}}

\title{Sequential Consistency and Concurrent Data Structures}
\titlerunning{SC Data Structures}

\author[1]{Ali Sezgin}
\affil[1]{University of Cambridge, UK\\ \texttt{as2418@cam.ac.uk}}
\authorrunning{A. Sezgin}

\Copyright{Ali Sezgin}

\serieslogo{}
\volumeinfo
  {Billy Editor and Bill Editors}
  {2}
  {Conference title on which this volume is based on}
  {1}
  {1}
  {1}
\EventShortName{}
\DOI{10.4230/LIPIcs.xxx.yyy.p}

\MYCOMMENT{
}

\begin{document}

\maketitle

\begin{abstract}
Linearizability, the de facto correctness condition for concurrent data structure implementations, despite its intuitive appeal is known to lead to poor scalability.
This disadvantage has led researchers to design scalable data structures satisfying consistency conditions weaker than linearizability.
Despite this recent trend, sequential consistency as a strictly weaker consistency condition than linearizability has received no interest.

In this paper, we investigate the applicability of sequential consistency as an alternative correctness criterion for concurrent data structure implementations.
Our first finding formally justifies the reluctance in moving towards sequentially consistent data structures: 
Implementations in which each thread modifies only its thread-local variables are sequentially consistent for various standard data structures such as pools, queues and stacks.
We also show that for almost all data structures, and all the data structures we consider in this paper, it is possible to have sequentially consistent behaviors in which a designated thread does not synchronize at all.
As a potential remedy, we define a hierarchy of quantitatively strengthened variants of sequential consistency such that the stronger the variant the more synchronization it enforces which at the limit is equal to that enforced by linearizability.
\end{abstract}


\section{Introduction}
\mylabel{sec:introduction}
The tension between performance and correctness is well known when it comes to developing low-level library routines implementing concurrent data structures~\cite{Sha2011}.
On the one hand, scalability, the ability to fully utilize the parallelism offered by the underlying architecture and generally accepted to be the main determinant of overall performance, is adversely affected by the need to synchronize among threads.
On the other hand, correctness criteria, usually known as consistency conditions, enforce lower bounds on the amount of synchronization.

In order to break the impasse in favor of scalability research has focused on weakening the notion of correctness.
The move towards alternative consistency conditions potentially leading to better performance was initially in the domain of memory implementations. 
The goal was to replace sequential consistency~\cite{Lam1979} (SC) with weaker memory consistency models (e.g.~\cite{ABJ+1993,ANB+1995,Goo1991,HS1992,LS1988}).
Recently a similar surge has been going on relative to linearizability~\cite{HW1990}, which has hitherto been the criterion for correctness in the domain of concurrent data structures, being replaced with weaker notions of correctness (e.g.~\cite{BGY+2014,HKP+2013,JR2014}).
There have been already many scalable implementations benefiting from these weaker notions (e.g.~\cite{AKY2010,AKL+2015,HHK+2013,HKP+2013,KLP2013,RSD2014}).

In this paper, we investigate SC as an alternative relaxation of linearizability.
Linearizability requires that the effect of a method be globally visible (i.e. to other executing threads) before it completes.
SC relaxes this by requiring only thread local ordering. 
That is, if two methods $m$ and $m'$ are called by the same thread in that order, then the requirement by SC is that the effect of $m$ {\em appear to} be before that of $m'$.
Furthermore, because SC is generally accepted to be the most intuitive memory consistency model and concomitantly is very well understood and studied, it is surprising that it has received no consideration until now.
Intriguing though it may be, we show that the apparent reluctance is actually warranted.

Our investigation of SC ranges over five common data structures: two variants of pools ($\mypool$, $\mypoolmem$), queues ($\myqueue$), stacks ($\mystack$) and register (banks) ($\myreg$).
We show that for all five of them it is possible to construct SC implementations where one thread, say $t$, can arbitrarily delay its synchronization with other threads as long as the sequence of local events of $t$ satisfy a certain property, which we call {\em robustness}.
Basically a sequence of events, method calls with return values, over a data structure $\mathcal D$ is robust if the same sequence can be executed regardless of the state $\mathcal D$ is in.
For instance, enqueueing or pushing an element is a robust sequence (of length 1) in $\myqueue$ and $\mystack$, respectively.

We also show that for pool, queue and stack implementations being SC guarantees even less.
We define the class of {\em composable} data structures to which pool, queue and stack belong.
Intuitively a data structure is composable if for any pair of valid behaviors there exists an interleaving of this pair which is again a valid behavior. 
For composable data structures, an implementation in which threads do not synchronize at all is SC.
To better understand the strength of such a result, consider a typical concurrency programming problem which contains $N$ tasks to be generated by the producer threads and to be completed by the consumer threads.
If the producers are to convey their tasks to the consumers over an SC queue (pool or stack), due to lack of synchronization the program will end with the queues of producers containing tasks and consumers having done nothing at all!

As a possible remedy, we propose a natural modification to the definition of SC.
We call an implementation $k$-SC if every thread has to synchronize at least once after executing $k$ local events. 
This definition is strong enough to rule out the pathological implementations we consider in this paper.
It also naturally spans the domain of implementations between SC and linearizability via a quantitative stratification.
The smaller $k$ is, the stronger $k$-SC is which at the limit reduces to linearizability (equivalent to $0$-SC).

To summarize, we make the following contributions:
\begin{itemize}
\item Define the properties robustness and composability for data structures,
\item Prove that essentially broken SC implementations exist for data structures with either of these properties,
\item Span the range between SC and linearizability by bounding the non-synchronized event sequences.
\end{itemize}

\subsection{Related Work}
\mylabel{subsec:related-work}
Directly related to our work, there have been two other work on relaxing the notion of linearizability.
In~\cite{HKP+2013}, Henzinger et al. propose a framework in which it is possible to quantitatively relax any sequential data structure.
Their framework enables one to define a desired metric and for any data structure $\mathcal D$ defines $\mathcal{D}_k$ to be all behaviors that are at most $k$ away from some valid behavior of $\mathcal{D}$.
A concurrent behavior is associated with the set of potential sequential witnesses, as defined by linearizability, and the concurrent behavior is correct if at least one potential witness belongs to $\mathcal{D}_k$.
In contrast, our work modifies the set of potential sequential witnesses, by using SC rather than linearizability, and the concurrent behavior is correct if at least one potential witness from this extended set belongs to $\mathcal{D}$.
In~\cite{JR2014}, Jagadeesan and Riely consider quiescent consistency (QC) as an alternative relaxation to linearizability. 
They quantitatively span the range between QC and linearizability.
Similar to our work, their quantitative metrics is also defined over concurrent behaviors.
Unlike us, they do not consider whether QC allows pathological cases.
Neither work formally establishes the relation between synchronization and the characteristics of data structures as we do in this work.


\section{Notation}
\mylabel{sec:notation}
For a set $A$, let $Pow(A)$ denote the set of all subsets of $A$.
Let $\lambda x.0$ with $x$ ranging over elements of $A$ denote the function that maps all elements of $A$ to 0.
Let $A[B]$ denote the collection of functions from $A$ to $B$.
For $f\in A[B]$, $f[a\mapsto b]$ denotes the function that agrees with $f$ on $A$ except for $a$ which is mapped to $b$.
A sequence $a$ of length $n$ over some alphabet $A$ is denoted by $a(1)\cdot a(2)\ldots \cdot a(n)$.
Alternatively, we also use the notation $\langle a(i)\rangle_{i\in[1,n]}$ to denote $a$.
Let $len(a)$ denote the length of $a$.
For simplicity of presentation, unless stated explicitly to be otherwise, for every sequence $a$ and for any $1\leq i,j\leq len(a)$ we assume that $a(i)\neq a(j)$.
Let $last(a)$ denote the last symbol of $a$; i.e. $last(a)=a(len(a))$.
Let $Set(a)$ denote the set of symbols occurring in $a$; i.e. $Set(a)=\{ a(i) \mid i\in[1,len(a)] \}$.
Each sequence $a$ induces a total order over $Set(a)$, {\em appears-before} order $\abOrder{a}$, such that $i<j$ iff $a(i)\abOrder{a} a(j)$.
Let $A^*$ denote the set of all sequences over $A$, with $\varepsilon$ denoting the {\em empty} sequence.

A labelled transition system (LTS) is a tuple $\LTS=(Q,q_0,L,\rightarrow)$, where $Q$ is the set of {\em states}, $q_0$ is the {\em initial} state, $L$ is a set of labels, and $\rightarrow\subseteq (Q\times L\times Q)$ is a {\em transition relation}.
We write $q\xrightarrow{l}q'$ if $(q,l,q')\in\rightarrow$.
A {\em run} $\mathbf{r}=q_0\cdot l_1\cdot q_1\dots l_n\cdot q_n$ is an alternating sequence of states and labels such that for all $i\in[1,n]$ we have $q_{i-1}\xrightarrow{l_i}q_i$.
The {\em trace} of $\mathbf{r}$, $\mytracex {\mathbf{r}}$, is the sequence of labels occurring in $\mathbf{r}$; i.e. $\mytracex {\mathbf{r}} = \mycompact l n$.
Let $\myTrx \LTS$ denote the set of all traces of $\LTS$.

%

\subsection{Data Structures}
\label{subsec:data-structures}
A {\em data structure} $\mathcal{D}$ is a pair $(D,\Sigma_{\mathcal{D}})$, where $D$ is the {\em data domain} and $\Sigma_{\mathcal{D}}$ is the {\em method alphabet}.
For all the data structures we consider in this paper, we take $D$ to be the set of natural numbers, $\mathbb{N}$, possibly augmented with a distinguished symbol $\NULL$.
An event of $\mathcal{D}$ is a quadruple $(id,m,d_i,d_o)$, where $id\in\mathbb{N}$ is an {\em event identifier}, $m\in \Sigma_{\mathcal{D}}$ is a method, $d_i,d_o\in D$ are {\em input} and {\em output} arguments, respectively.
Intuitively, $(id, m, d_i, d_o)$ denotes the application of method $m$ with input argument $d_i$ returning the output value $d_o$.
When the input (resp. output) argument is not used in the event, we write $(id,m,\bot,d_o)$ (resp. $(id,m,d_i,\bot)$). 
We will assume that each event has a unique event identifier.
We will use $\Evt {\mathcal{D}}$ to denote the set of all events of $\mathcal{D}$.
A duplicate-free sequence over $\Evt {\mathcal{D}}$ is called a $\mathcal{D}$-behavior.
The {\em semantics} of data structure $\mathcal{D}$ is a set of $\mathcal{D}$-behaviors, each of which is called a {\em valid} behavior.
For each data structure $\mathcal{D}$, we will define a labelled transition system $\LTSx {\mathcal{D}}$ such that $\mathbf{e}\in \myTrx {\LTSx{\mathcal{D}}}$ iff $\mathbf{e}$ is a valid $\mathcal{D}$-behavior.
Below we list the data structures that we will consider in this paper.

\subsubsection{Pool, $\mypool$}
The method alphabet $\Sigma_{\mypool}$ of a pool is the set $\{\myput,\mytake\}$. 
Events of $\mypool$ are written as $\myput^{id}(x)$, short for $(id,\myput,x,\bot)$, and $\mytake^{id}(x)$, short for $(id,\mytake,\bot,x)$.
For conciseness, from this point on we will omit the superscript $id$.
Events with $\myput$ are called {\em put} events, and those with $\mytake$ are called {\em take} events.
We use $\myPut$ and $\myTake$ to denote the set of of all put and take events, respectively.

$\LTSx \mypool$ is defined as $(Pow(\mathbb{N}),\emptyset, \Evt{\mypool}, \rightarrow_{\mypool})$, where $\rightarrow_{\mypool}$ is defined as:
\begin{itemize}
\item $q\xrightarrow{\myputx x}_{\mypool} q'$ iff $q'=q\cup\{x\}$,
\item $q\xrightarrow{\mytakex x}_{\mypool}q'$ iff $x\in q$ and $q'=q\setminus\{x\}$,
\item $q\xrightarrow{\mytakex \NULL}_{\mypool}q'$ iff $q=q'=\emptyset$.
\end{itemize}

\subsubsection{Pool with Membership, $\mypoolmem$}
The method alphabet $\Sigma_{\mypoolmem}$ is $\Sigma_{\mypool}\cup \{\mymem\}$.
An event of $\mypoolmem$ is either a $\mypool$ event or of the form $\mymem(x,y)$, short for $(id,\mymem,x,y)$.
Events with $\mymem$ are called {\em query} events, and $\myMem$ denotes the set of all query events.

$\LTSx \mypoolmem$ is defined as $(Pow(\mathbb{N}),\emptyset,\Evt{\mypoolmem},\rightarrow_{\mypoolmem})$, where $\rightarrow_{\mypoolmem}$ is defined as:
\begin{itemize}
\item $q\xrightarrow{x}_{\mypoolmem} q'$ if $q\xrightarrow{x}_{\mypool}q'$,
\item $q\xrightarrow{\mymemx x y}_{\mypoolmem} q'$ iff $q=q'$, and either $y=x$ and $x\in q$, or $y=x+1$ and $x\notin q$.
\end{itemize}

\subsubsection{Queue, $\myqueue$}
The method alphabet $\Sigma_{\myqueue}$ is the set $\{\myenq,\mydeq\}$.
Events of {\myqueue} are written as $\myenq(x)$, short for $(id,\myenq,x,\bot)$, and $\mydeq(x)$, short for $(id,\mydeq,\bot,x)$.
Events with $\myenq$ are called {\em enqueue} events, and those with $\mydeq$ are called {\em dequeue} events.
We use $\myEnq$ and $\myDeq$ to denote the set of all enqueue and dequeue events, respectively.

$\LTSx \myqueue$ is defined as $(\mathbb{N}^*,\varepsilon,\Evt{\myqueue},\rightarrow_{\myqueue})$, where $\rightarrow_{\myqueue}$ is defined as:
\begin{itemize}
\item $q\xrightarrow{\myenqx x}_{\myqueue}q'$ iff $q'=q\cdot x$,
\item $q\xrightarrow{\mydeqx x}_{\myqueue}q'$ and $x\neq\NULL$ iff $q=x\cdot q'$,
\item $q\xrightarrow{\mydeqx \NULL}_{\myqueue}q'$ iff $q=q'=\varepsilon$.
\end{itemize}

\subsubsection{Stack, $\mystack$}
The method alphabet $\Sigma_{\mystack}$ is the set $\{\mypush,\mypop\}$.
Events of $\mystack$ are written as $\mypush(x)$, short for $(id,\mypush,x,\bot)$, and $\mypop(x)$, short for $(id,\mypop,\bot,x)$.
Events with $\mypush$ are called {\em push} events, and those with $\mypop$ are called {\em pop} events.
We use $\myPush$ and $\myPop$ to denote the set of all push and pop events, respectively.

$\LTSx \mystack$ is defined as $(\mathbb{N}^*,\varepsilon,\Evt{\mystack},\rightarrow_{\mystack})$, where $\rightarrow_{\mystack}$ is defined as:
\begin{itemize}
\item $q\xrightarrow{\mypushx x}_{\mystack}q'$ iff $q'=q\cdot x$,
\item $q\xrightarrow{\mypopx x}_{\mystack}q'$ and $x\neq\NULL$ iff $q=q'\cdot x$,
\item $q\xrightarrow{\mypopx \NULL}_{\mystack}q'$ iff $q=q'=\varepsilon$.
\end{itemize}

\subsubsection{Register, $\myreg$}
The method alphabet $\Sigma_{\myreg}$ is the set $\{\mywr_i,\myrd_i \mid i \in \mathbb{N}\}$.
Events of $\myreg$ are written as $\mywr_i(x)$, short for $(id,\mywr_i,x,\bot)$, and $\myrd_i(x)$, short for $(id,\myrd_i,\bot,x)$.
Events with $\mywr_i$ are called {\em write} events, and those with $\myrd_i$ are called {\em read} events.
We use $\myWr$ and $\myRd$ to denote the set of all write and read events, respectively.

$\LTSx \myreg$ is defined as $(\mathbb{N}[\mathbb{N}],\lambda x.0,\Evt{\myreg},\rightarrow_{\myreg})$, where $\rightarrow_{\myreg}$ is defined as:
\begin{itemize}
\item $q\xrightarrow{\mywrx i x}_{\myreg}q'$ iff $q'=q[i\mapsto x]$.
\item $q\xrightarrow{\myrdx i x}_{\myreg}q'$ iff $q'=q$ and $q(i)=x$.
\end{itemize}

\subsection{Histories}
\label{subsec:histories}
Each event $e=(uid,m,d_i,d_o)$ generates two {\em actions}: the {\em invocation} of $e$, written as $inv(e)$, and the {\em response} of $e$, written as $res(e)$.
We will also use $m_i^{uid}(d_i)$ and $m_r^{uid}(d_o)$ to denote the invocation and response actions, respectively.
When a particular method $m$ does not have an input (resp. output) parameter, we will write $m_i^{uid}$ (resp. $m_r^{uid}$) for the corresponding invocation (resp. response) action. 
We will also often omit the superscripts, when they are not important.
For an event set $E$, let $E_{i}$ and $E_{r}$ denote the set of all invocation and response actions generated by $E$.

A {\em $\mathcal{D}$-history} is a sequence of invocation and response actions generated by $E_{\mathcal{D}}$.
The unique identifier of each event unambigiously pairs each invocation action to a unique response action; in such a case, the actions are said to {\em match}.
We will make use of this pairing without explicitly referring to event identifiers when there is no confusion.
Similarly, we will omit $\mathcal{D}$ whenever it is either inconsequential or clear from the text.
A history $\history{h}$ is {\em well-formed} if every response action appears after its matching invocation action in $\history{h}$.
An event $e$ is {\em completed} in $\history{h}$, if both of its invocation and response actions appear in $\history{h}$. 
Formally, $e$ is completed if $e_i,e_r\in Set({\history{h}})$.
A history $\history{h}$ is {\em complete} if for all events $e$, $e_i\in Set(\history{h})$ iff $e_r\in Set(\history{h})$.
In what follows we will consider only well-formed and complete histories.

An event $e$ precedes another event $e'$ in $\history{h}$, written $e\precOrder {\history{h}} e'$, if the response action of $e$ appear before the invocation action of $e'$; i.e. $e_r\abOrder {\history{h}} e'_i$.
A history is called {\em sequential} if all invocation actions are immediately followed by their matching responses. 
Formally, the (complete) history $\history{h}$ is sequential if it is of the form $e_{1,i}\cdot e_{1,r}\cdot\ldots\cdot e_{n,i}\cdot e_{n,r}$.
We identify sequential $\mathcal{D}$-histories with $\mathcal{D}$-behaviors by mapping each matching pair of invocation and response actions to the event generating them.
A sequential history $\history{s}$ is a {\em linearization} of a history $\history{h}$, if $\history{s}$ is a permutation of $\history{h}$ such that $e\precOrder{\history{h}}e'$ implies $e\precOrder{\history{s}}e'$.

\begin{definition}[Linearizability]
A $\mathcal{D}$-history $\history{h}$ is linearizable if there exists a linearization of $\history{h}$ that is a legal $\mathcal{D}$-behavior.
A set $H$ of histories is linearizable if every $\history{h}\in H$ is linearizable.
\end{definition}

Let $T$ be the set of {\em thread id}'s.
A threaded $\mathcal{D}$-action is of the form $(t,e)$, where $t\in T$ and $e\in E_{\mathcal{D},i}\cup E_{\mathcal{D},r}$.
For a threaded action $\vec{a}=(t,a)$, we use $tid$ and $act$ to retrieve the first and second components, respectively; i.e. $tid(\vec{a})=t$ and $act(\vec{a})=a$.
Similarly, $tid$ and $act$ are point-wise extended over sequences of threaded-actions. 
For a sequence of threaded-actions $\thrhist{h}$ and a thread id $t\in T$, let $\thrhist{h}\downarrow_t$ denote the subsequence obtained by removing all threaded-actions from $\thrhist{h}$ whose first component is not equal to $t$.
A {\em threaded} $\mathcal{D}$-history is a sequence $\thrhist{h}$ over threaded actions such that 
\begin{itemize}
\item the sequence $act(\thrhist{h})$ is a complete $\mathcal{D}$-history,
\item the sequence $act(\thrhist{h}\downarrow_{t})$ is sequential for any $t\in T$.
\end{itemize}
The second condition implies that at any point in a threaded history any thread $t\in T$ can have at most one unmatched action.
For ease of presentation, we will use $\history{h}$ and $\history{h}(t)$ to denote $act(\thrhist{h})$ and $act(\thrhist{h}\downarrow_{t})$, respectively.
We will extend the properties of histories to threaded histories: a threaded history $\thrhist{h}$ is said to satisfy property $P$ if $\history{h}$ satisfies $P$.

\begin{definition}[Sequential Consistency]
A threaded $\mathcal{D}$-history is sequentially consistent if there is a sequential threaded $\mathcal{D}$-history $\thrhist{s}$ such that $\thrhist{s}$ is a permutation of $\thrhist{h}$ and for all $t\in T$ we have $\history{s}(t)=\history{h}(t)$.
\end{definition}
Intuitively, for a threaded history $\thrhist{h}$ to be sequentially consistent (SC) only the relative ordering per thread has to be respected.
Since by the definition of threaded-histories, if $(t,e)$ and $(t,e')$ are both in $\thrhist{h}$, then either $e\precOrder{\history{h}}e'$ or $e'\precOrder{\history{h}}e$, linearizability must preserve the relative ordering among events done by the same thread. 
In other words, linearizability is a stronger condition than sequential consistency.

\begin{fact}\mylabel{fact:sc-lin}
Let $\thrhist{h}$ be a threaded $\mathcal{D}$-history. 
It is sequentially consistent if it is linearizable.
\end{fact}

As far as the implementation of data structures is concerned, we will not specify a particular programming language.
We will assume that each method in the method alphabet has an accompanying procedure.
An {\em execution trace} is a sequence of instruction labels coupled with thread identifiers executing the instruction.
For instance, $(t:i)$ denotes the execution of instruction with the unique label $i$ by thread $t$.
An instruction label is the {\em entry} point of method $m$, written $enter(m)$, if it is the label of the first instruction of $m$.
Similarly, an instruction label is an {\em exit} point of $m$, written $exit(m)$, if it is the lable of an instruction that completes the execution of $m$.
Each execution trace $\tau$ induces a history $h(\tau)$ which is obtained by replacing each $(t : enter(m))$ with $m_i^{t_{uid}}(d_i)$, each $(t : exit(m))$ with $m_r^{t_{uid}}(d_o)$, and removing the remaining (intermediate) symbols. 
We assume that states of an execution trace contain enough information to deduce the values of $d_i$ and $d_o$ associated with each entry and exit point.
An execution trace is {\em complete} if its induced history is complete.
For an execution trace $\tau$ and some $t\in T$, let $\tau\downarrow_{t}$ denote the execution trace obtained by retaining only the symbols due to $t$ (symbols of the form $(t : i)$.
An implementation is identified with the set of execution traces it generates.
When clear from the context we will refer to the induced history of an execution trace as a history of the implementation.

\section{Properties of Data Structures}
\mylabel{sec:data-structure-properties}
Our main result is that sequential consistency is too weak because the class of SC implementations includes {\em bad} ones.
In order to generalize our result, we abstract away irrelevant specifics of data structures and extract what seems to be the essential property that causes SC  implementations to misbehave.
We identify two properties: composability and robustness.
A data structure is composable if any two valid behavior can be interleaved in such a way that the result is also a valid behavior.
Robustness means that the data structure has events that are state independent.
In this section, we formalize these notions.  

\begin{definition}[Composable]
Let $\history{e}$ and $\history{f}$ be two valid $\mathcal{D}$-behaviors.
They are called {\em composable} if there exist two partitionings $\history{e}=\history{e}_1\ldots\history{e}_k$ and $\history{f}=\history{f}_1\ldots\history{f}_k$ such that the behavior $\history{e}_1\history{f}_1\ldots\history{e}_k\history{f}_k$ is a valid $\mathcal{D}$-behavior. 
The data structure $\mathcal{D}$ is called {\em composable} if any pair of valid $\mathcal{D}$-behaviors are composable.
\end{definition}
Informally, two valid $\mathcal{D}$-behaviors are composable if it is possible to interleave them to obtain another valid $\mathcal{D}$-behavior. 
We now proceed with a series of results, establishing composability of the data structures we consider in this paper.

\begin{lemma}[Pool and Composability]
The pool data structure $\mypool$ is composable.
\end{lemma}
\begin{proof}
Let $\history{e}$ and $\history{f}$ be two valid $\mypool$-sequences.
Let $\history{e}_1$ and $\history{f}_1$ be the maximal prefixes of $\history{e}$ and $\history{f}$, respectively, such that $last(\history{e}_1)=\mytakex \NULL$ and $last(\history{f}_1)=\mytakex \NULL$. 
Let $\history{e}_2$ and $\history{f}_2$ be the remaining suffixes of $\history{e}$ and $\history{f}$.
That is, $\history{e}=\history{e}_1\history{e}_2$ and $\history{f}=\history{f}_1\history{f}_2$. 
We claim that $\history{g}=\history{e}_1\history{f}_1\history{e}_2\history{f}_2$ is a valid $\mypool$-behavior.
This is equivalent to showing that there is a run of $\LTSx {\mypool}$ whose trace is $\history{g}$.
By the validity of $\history{e}$, we know that there is a run $\mathbf{r}_1$ of $\LTSx {\mypool}$ with trace $\history{e}_1$ because $\myTrx {\LTSx \mypool}$ is prefix-closed.
Furthermore, by the assumption that $last(\history{e}_1)=\mytakex \NULL$, that run ends at the state $\emptyset$.
Similarly, there is a run $\mathbf{r}_2$ with trace $\history{f}_1$.
Since $\mathbf{r}_1$ ends at $\emptyset$, $\mathbf{r}_1\mathbf{r}_2$ is also a run of $\LTSx {\mypool}$ whose trace is $\history{e}_1\history{f}_1$.
Since the trace $\history{e}_2$ belongs to a run $\mathbf{r}_3$ (implying that $\mathbf{r}_1\mathbf{r}_3$ is the run with trace $\history{e}$) which starts at $\emptyset$, $\mathbf{r}_1\mathbf{r}_2\mathbf{r}_3$ is a run of $\LTSx {\mypool}$.
Finally, we have to extend $\mathbf{r}_1\mathbf{r}_2\mathbf{r}_3$ with a run $\mathbf{r}_4$ whose trace is $\history{f}_2$.
We proceed by induction on the length of $\history{f}_2$.
The base case, $\history{f}_2=\varepsilon$ is trivial.
Assume that the next event is $\myputx x$. 
By the definition of $\LTSx {\mypool}$, such a transition is enabled at all states.
Assume that the next event is $\mytakex x$ (note that by construction, $\history{f}_2$ does not contain a transition with label $\mytakex {\NULL}$).
Because the run with trace $\history{f}_2$ starts at the state $\emptyset$, there must have been an event $\myputx x$ in $\history{f}_2$ and that no event since the last occurrence in $\history{f}_2$ prior to $\mytakex x$ was equal to $\mytakex x$.
In other words, we must have $x\in q$, where $q$ is the current state, which means that it is possible to extend the run with $\mytakex x$.
\end{proof}

\begin{lemma}[Stack and Composability]
The stack data structure $\mystack$ is composable.
\end{lemma}
\begin{proof}[Proof (Sketch)]
The construction is similar to the one given in the proof of $\mypool$.
For any two valid $\mystack$-behaviors $\history{e}$ and $\history{f}$, we take the maximal prefixes $\history{e}_1$ and $\history{f}_1$ both of which end with $\mypopx {\NULL}$. 
Then, the composed behavior $\history{e}_1\history{f}_1\history{e}_2\history{f}_2$ is also a valid $\mystack$-behavior, where $\history{e}_2$ and $\history{f}_2$ are the remaining suffixes of $\history{e}$ and $\history{f}$.
Intuitively, the constructed behavior is valid because neither $\history{e}_2$ nor $\history{f}_2$ reaches beyond what they have pushed onto the stack (no appearance of $\mypopx {\NULL}$ in either of the two). 
\end{proof}

\begin{lemma}[Queue and Composability]
The queue data structure $\myqueue$ is composable.
\end{lemma}
\begin{proof}[Proof (Sketch)]
Unlike the previous two cases, the construction for queue is more involved.
We will need to partition a valid $\myqueue$-behavior according to the relative ordering between the enqueue events of $\myEnq$ and the dequeue events of $\myDeq$.
For any $\myenqx x$, let $\mydeqx x$ be called its {\em observer}.
Given a valid $\myqueue$-behavior $\history{e}$, we define $Era_j(\history{e})$ inductively as follows:
\begin{itemize}
\item $Era_0(\history{e})=\varepsilon$.
\item $Era_{j+1}(\history{e})$ the maximum segment that begins with an enqueue event whose observer is in $Era_{j}(\history{e})$ and extends until the first element that belongs to $Era_{j}(\history{e})$.
\end{itemize}
To illustrate the definition, consider the valid $\myqueue$-behavior:
\[
\history{e} \quad = \quad \myenqx 1\cdot \myenqx 2\cdot \mydeqx 1 \cdot \myenqx 3 \cdot \mydeqx 2 \cdot \myenqx 4 \cdot \myenqx 5 \cdot \mydeqx 3 \cdot \mydeqx 4 
\]
Then, we have 
\begin{eqnarray*}
Era_1(\history{e}) &=& \myenqx 5 \cdot \mydeqx 3 \cdot \mydeqx 4\\
Era_2(\history{e}) &=& \myenqx 3 \cdot \mydeqx 2 \cdot \myenqx 4\\
Era_3(\history{e}) &=& \myenqx 2 \cdot \mydeqx 1\\
Era_4(\history{e}) &=& \myenqx 1
\end{eqnarray*}
Observe that by construction all enqueue events of $Era_j(\history{e})$ for $j>1$ are observed by the dequeue events of $Era_{j-1}(\history{e})$. 
By convention, for any $n>k$ such that $Era_k(\history{e})\ldots Era_1(\history{e}) = \history{e}$, we set $Era_n(\history{e})=\varepsilon$.

Now let $\history{e}$ and $\history{f}$ be two valid $\myqueue$-behaviors.
Let $\history{e}_1$ be the maximal prefix of $\history{e}$ such that the run with trace $\history{e}_1$, which necessarily exists, ends at state $\varepsilon$ (the initial state of $\LTSx {\myqueue}$).
Let $\history{e}_r$ be the remaining suffix of $\history{e}$; i.e. $\history{e}=\history{e}_1\history{e}_r$.
Let $\history{f}_1$ and $\history{f}_r$ be defined similarly for $\history{f}$.
Now construct the $Era$ sequences for $\history{e}_r$ and $\history{f}_r$. 
Let $j_e$ be the maximal index of a non-empty era sequence of $\history{e}$, $j_f$ be the index for $\history{f}$.
Without loss of generality, assume that $j_e\geq j_f$.
Then, the interleaving
\[
\history{e}_1\cdot \history{f}_1\cdot Era_{e_j}(\history{e})\cdot Era_{e_j}(\history{f}) \ldots Era_1(\history{e})\cdot Era_1(\history{f})
\]
is a valid {\myqueue}-behavior.
To illustrate the construction, consider another valid $\myqueue$-behavior
\[
\history{f} \quad = \quad \myenqx 6 \cdot \mydeqx 6 \cdot \myenqx 7 \cdot \myenqx 8 \cdot \mydeqx 7 \cdot \myenqx 9 \cdot \myenqx {10} \cdot \mydeqx 8
\]
Then, $\history{f}_1$ is $\myenqx 6 \cdot \mydeqx 6$ and the era sequences for the suffix $\history{f}_r$ is
\begin{eqnarray*}
Era_1(\history{f}) &=& \myenqx 9 \cdot \myenqx {10} \cdot \mydeqx 8\\
Era_2(\history{f}) &=& \myenqx 8 \cdot \mydeqx 7\\
Era_3(\history{f}) &=& \myenqx 7\\
\end{eqnarray*}
Finally, the interleaving for $\history{e}$ and $\history{f}$ is given as:
\begin{multline*}
\myenqx 6 \cdot \mydeqx 6 \cdot \myenqx 1 \cdot \myenqx 2 \cdot \mydeqx 1 \cdot \myenqx 7 \cdot \myenqx 3 \cdot \mydeqx 2 \cdot \\
\myenqx 4 \cdot \myenqx 8 \cdot \mydeqx 7 \cdot \myenqx 5 \cdot \mydeqx 3 \cdot \mydeqx 4 \cdot \myenqx 9 \cdot \myenqx {10} \cdot \mydeqx 8
\end{multline*}
which is a valid $\myqueue$-behavior.
Intuitively, the construction works as each era sequence of one behavior (say $\history{e}$) with index $j$ removes all the elements inserted by the most recent era sequence of the same behavior with index $j+1$, thereby making the insertions of this behavior ($\history{e}$) invisible to the other behavior ($\history{f}$) and vice versa. 
\end{proof}
The other two remaining data structures, $\mypoolmem$ and $\myreg$, are not composable.
For the pool with membership data structure $\mypoolmem$, the following two valid $\mypoolmem$-behaviors have no interleaving that is a valid $\mypoolmem$-behavior:
\[
\history{e} = \myputx 1 \cdot \mymemx 2 3, \quad \history{f} = \myputx 2 \cdot \mymemx 1 2
\]
since $\mymemx 2 3$ has to come before $\myputx 2$ and $\mymemx 1 2$ has to come before $\myputx 1$, both conditions of which cannot be simultaneously satisfied.

Similarly, for the register data structure $\myreg$, the following valid $\myreg$-behaviors are not composable:
\[
\history{e} = \mywrx 1 1 \cdot \myrdx 2 0, \quad \history{f} = \mywrx 2 1 \cdot \myrdx 1 0
\]
The next property, robustness, is satisfied by all five of the data structures we consider and to the best of our knowledge all data structures are robust.
\begin{definition}[Robustness]
Let $\history{e}$ be a valid $\mathcal{D}$-behavior.
It is {\em robust} if for any state $q\in \LTSx {\mathcal D}$, there is a run that starts at $q$ with trace $\history{e}$.
A data structure is robust if it contains at least one robust sequence.
\end{definition}
In general, any data structure which contains a {\em total} event (e.g. $\myenqx x$ of $\myqueue$ or $\mypushx x$ of $\mystack$) is robust.
\begin{lemma}\mylabel{lem:all-singular}
All of $\mypool$, $\mypoolmem$, $\myqueue$, $\mystack$, $\myreg$ are robust.
\end{lemma}
\begin{proof}
The event $\mytakex x$ is enabled at every state of $\LTSx {\mypool}$ and $\LTSx {\mypoolmem}$.
The event $\myenqx x$ is enabled at every state of $\LTSx {\myqueue}$.
The event $\mypushx x$ is enabled at every state of $\LTSx {\mystack}$.
The event $\mywrx x y$ is enabled at every state of $\LTSx {\myreg}$.
\end{proof}
Robust sequences can contain arbitrary events or be restricted to a subset of events, depending on the data structure.
\begin{lemma}
A robust sequence of $\mypool$ and $\mypoolmem$ cannot contain $\myputx \NULL$; for $\myqueue$, it cannot contain $\mydeqx \NULL$; for $\mystack$, it cannot contain $\mypopx \NULL$. 
There is no restriction on robust sequences for $\myreg$.
\end{lemma}

\newcommand{\myimpliso}{\ensuremath{\mathsf{Imp}_{iso}}}
\newcommand{\myimplisox}[1]{\ensuremath{\myimpliso(#1)}}
\newcommand{\myimplclose}{\ensuremath{\mathsf{Imp}_{sing}}}
\newcommand{\myimplclosex}[1]{\ensuremath{\myimplclose(#1)}}

\newcommand{\myObjx}[1]{\ensuremath{\mathtt{Obj}_{#1}}}

\section{SC is too Weak}
\mylabel{sec:sc-weak}
In this section, we present the bad implementations that SC seems to allow.
These bad implementations come in two variants: conditional and unconditional non-synchronization.
We show that all robust data structures allow conditional non-sychronization.
Unconditional non-synchronization, arguably the worst of the two, is allowed by composable data structures ($\mypool,\myqueue,\mystack$).

Let us call a label $l$ {\em enabled} at state $q$ if there exists a state $q'$ such that $q\xrightarrow{l}q'$.
\begin{definition}[Initialized]
Let $(Q,q_0,L,\rightarrow)$ be an LTS.
A state $q\in Q$ is {\em subsumed} by another state $q'\in Q$ if for all $l\in L$, $l$ is enabled at $q$ implies $l$ is enabled at $q'$.  
The LTS is called {\em initialized} if its initial state $q_0$ is not subsumed by any other state.
\end{definition}
We call an LTS $(Q,q_0,L,\rightarrow)$ {\em non-trivial} if there is at least one state $q\neq q_0$ and one label $l\in L$ such that $q_0\xrightarrow{l}q$.
A data structure $\mathcal{D}$ is non-trivial if $\LTSx {\mathcal{D}}$ is non-trivial.

\begin{lemma}\mylabel{lem:all-initialized}
The LTS corresponding to $\mypool,\mypoolmem,\myqueue,\mystack,\myreg$ are initialized.
\end{lemma}
\begin{proof}
In all cases, there are transitions which are only enabled at the initial state.
The transitions for all except for $\LTSx \myreg$ are $\mytakex \NULL$ in $\LTSx \mypool$ and $\LTSx \mypoolmem$; $\mydeqx \NULL$ in $\LTSx \myqueue$; $\mypopx \NULL$ in $\LTSx \mystack$.
For $\LTSx \myreg$, let $q'\neq q_0$ be some state. 
By definition, there must be at least one $x\in \mathbb{N}$ such that $q'(x)\neq 0$ since otherwise $q'$ and $q_0=\lambda x.0$ are identical.
Then, $\myrdx x 0$ is enabled at $q_0$ but not at $q'$.
\end{proof}

For each data structure $\mathcal{D}$, we distinguish an implementation $\myimplisox {\mathcal{D}}$, called the {\em isolated implementation} of $\mathcal{D}$, whose induced histories are thread-locally valid $\mathcal{D}$-behaviors.
Formally, for any execution trace $\tau$ of $\myimplisox {\mathcal{D}}$ and for any $t\in Tid$, the induced history of $\tau\downarrow_{t}$ is a valid $\mathcal{D}$-behavior.
Intuitively, isolated implementations are those which do not need any communication between threads.
For most data structures, such implementations are not desirable and as the following result shows are ruled out by linearizability.

\begin{lemma}\mylabel{lem:initialized-nonlin}
If $\mathcal{D}$ is initialized and non-trivial, then $\myimplisox {\mathcal{D}}$ is not linearizable.
\end{lemma}
\begin{proof}
Let $\LTSx {\mathcal{D}}$ be the tuple $(Q,q_0,E,\rightarrow)$.
Because $\mathcal{D}$ is non-trivial, there is a transition $q_0\xrightarrow{e}q$ for some $e=(m,d_i,d_o)\in E$ and $q\neq q_0$.
Because $\mathcal{D}$ is initialized, there is an event $e'\in E$ such that $e'=(m',{d'}_i,{d'}_o)$ is enabled at $q_0$ and not at $q$.
Then consider the history $\history{h}$ for $t,t'\in T$:
\[
\history{h} \stackrel{def}{=} (t,m_i(d_i))\cdot (t,m_r(d_o)) \cdot (t',m_i'({d'}_i)) \cdot (t',m_r'({d'}_o))
\]
The only linearization for this history is $(m,d_i,d_o)\cdot (m',d'_i,d'_o)$.
Since this is not a valid $\mathcal{D}$-behavior, $\history{h}$ is not linearizable.
However, because both $(m,d_i,d_o)$ and $(m',d'_i,d'_o)$ are individually valid $\mathcal{D}$-behaviors, $\history{h}$ is induced by some execution trace of $\myimplisox {\mathcal{D}}$, implying that the latter is not linearizable.
\end{proof}
The following result immediately follows from Lemma's~\ref{lem:all-initialized} and~\ref{lem:initialized-nonlin}.
\begin{corollary}
The isolated implementations of $\mypool$, $\mypoolmem$, $\myqueue$, $\mystack$ and $\myreg$ are not linearizable.
\end{corollary}
The previous result shows that the definition of linearizability is strong enough to leave out these pathological implementations.
As we show next, sequential consistency is weak enough to allow for isolated implementations of some data structures.

\begin{theorem}[SC and Isolated Implementations]
If $\mathcal{D}$ is composable, then $\myimplisox {\mathcal{D}}$ is sequentially consistent.
\end{theorem}
\begin{proof}
Let $\tau$ be an execution trace of $\myimplisox {\mathcal{D}}$ and let $\thrhist{h}$ be the induced threaded-history.
We do induction on the number of threads that execute at least one method during $\tau$.
If $\thrhist{h}$ is due to a single thread, then it is a valid $\mathcal{D}$-behavior and hence SC.
Assume that if $\thrhist{h}$ has less than or equal to $k$ different threads, it is SC.
Now let $\thrhist{h}$ have $k+1$ different threads and let $t\neq t'$ be the identifiers of two of those.
By the definition of $\myimplisox {\mathcal{D}}$, both $\history{h}(t)$ and $\history{h}(t')$ are valid $\mathcal{D}$-behaviors.
By composability, there is an interleaving $\history{h}'$ of $\history{h}(t)$ and $\history{h}(t')$ which is also a valid $\mathcal{D}$-behavior.
Let $u\in T$ be a thread identifier that does not appear in $\thrhist{h}$ and let $\thrhist{h}'$ denote the threaded-history obtained by coupling each symbol of $\history{h}'$ with $u$.
That is, $\thrhist{h}'$ is the sequence $\mycompact {u,\history{h}'} {len(\history{h}')}$.
Let $\thrhist{g}$ be the threaded history which is constructed by first projecting out from $\thrhist{h}$ all symbols $s$ with $act(s)=t$ or $act(s)=t'$, and then extending it with $\thrhist{h}'$.
By construction, $\thrhist{g}$ has $k$ different threads and each $\history{g}(t)$ is a valid $\mathcal{D}$-behavior.
By definition, there is an execution trace $\tau'$ which induces $\thrhist{g}$.
By inductive hypothesis, there is a sequential threaded valid $\mathcal{D}$-history $\thrhist{s}$ such that for all $t\in T$ we have $\thrhist{s}(t)=\thrhist{g}(t)$.
Finally, because $\thrhist{h}'$ was an interleaving of $\thrhist{h}(t)$ and $\thrhist{h}(t')$, replacing $u$ in $\thrhist{s}$ with the original $t$ or $t'$ identifiers yields the valid sequential threaded $\mathcal{D}$-history corresponding to $\thrhist{h}$.
\end{proof}
This implies that the definition of sequential consistency is not strong enough for composable data structures.
\begin{corollary}
The isolated implementations of $\mypool$, $\myqueue$ and $\mystack$ are sequentially consistent.
\end{corollary} 

An execution trace $\tau$ of a $\mathcal{D}$ implementation is called {\em $t$-singular} if $h(\tau)\uparrow^t$ is linearizable and $h(\tau)\downarrow_t$ is a valid $\mathcal D$-behavior.
Let $\myimplclosex {\mathcal{D}}$, the {\em singular implementation} of $\mathcal{D}$, denote the union of all linearizable execution traces and all $t$-singular execution traces.
We next show that singular implementations of robust data structures are non-linearizable.

\begin{lemma}
If $\mathcal{D}$ is robust, initialized and non-trivial, then $\myimplclosex {\mathcal{D}}$ is not linearizable.
\end{lemma}
\begin{proof}
Let $\LTSx {\mathcal{D}}$ be the tuple $(Q,q_0,E,\rightarrow)$.
Let $\history{e}=\mycompact e n$ be a robust sequence.
Let $q_0\xrightarrow{e_1}q_1\ldots\xrightarrow{e_n}q_n$ be the run whose trace is $\history{e}$.
If $q_n\neq q_0$, because $\mathcal{D}$ is initialized, there must be some event $e'\in E$ enabled at $q_0$ but not at $q_n$.
Then, similar to the proof of Lemma~\ref{lem:initialized-nonlin}, the threaded history in which thread $t\in T$ runs all actions associated with $\history{h}$ followed by another thread $u\in T$ running the two actions generated by the event $e$ is in $\myimplclosex {\mathcal{D}}$ because $e$ represents a valid $\mathcal{D}$-behavior and $\history{e}$ represents a robust sequence.
If $q_n=q_0$, then there must be a state $q'\neq q_0$ and a label $e'$ such that $q_0\xrightarrow{e'}q'$ holds.
The rest of the argument is the same as the previous one.
\end{proof}
The following is immediate from Lemma~\ref{lem:all-singular}.
\begin{corollary}
The singular implementations of $\mypool$, $\mypoolmem$, $\myqueue$, $\mystack$ and $\myreg$ are not linearizable.
\end{corollary}
Intuitively, in singular implementations the behavior of some thread $t$ can be hidden from the rest of the system as long as the sequence generated by $t$ remains robust. 
This, although arguably not as bad as being isolated, is still an undesirable feature and linearizability forbids it.
We now show that singular implementations are sequentially consistent.
\begin{theorem}[SC and Singular Implementations]
For any data structure $\mathcal D$, $\myimplclosex {\mathcal D}$ is sequentially consistent.
\end{theorem}
\begin{proof}
Let $\tau$ be an execution trace of $\myimplclosex {\mathcal D}$ and let $\thrhist h$ be the induced threaded history.
If $\thrhist h$ is linearizable, then by Fact.~\ref{fact:sc-lin} it is also sequentially consistent.
If $\thrhist h$ is not linearizable, then there must be some thread $t\in T$ such that $h(\tau)\downarrow_t$ is robust.
Furthermore, we know that $h(\tau)\uparrow^t$ is linearizable.
Assume that $\history{s}$ is the linearization of $h(\tau)\uparrow^t$.
Then the sequence $\history{s}'$ formed by appending $h(\tau)\downarrow_t$ to $\history{s}$ is a valid $\mathcal{D}$-behavior. 
Since for all $u\in T$ we have $\history{s}'\downarrow_u=h(\tau)\downarrow_u$, we conclude that $\thrhist h$ is sequentially consistent.
\end{proof}

We end this section by giving templates for isolated and singular implementations.
Assume that we already have a sequential implementation of any data structure and use the notation $\myObjx {\mathcal D}$ to denote the class implementing the methods of $\mathcal{D}$.
For instance, for the $\myqueue$ data structure, $\myObjx {\myqueue}$ implements the required methods which are called by appending the method to an object $O$ of type $\myObjx {\myqueue}$ as in $O$.$\myenqx x$.
In our programs, we assume that each thread $t\in T$ has its thread-local copy of type $\myObjx {\mathcal D}$ and use the notation \texttt{Obj}[$t$] to denote the object exclusively used by $t$.
Then, in an isolated implementation, each method $m\in \Sigma_{\mathcal D}$ has the following template: 
\begin{figure}[hbt]
$m$($d_i$) \{ $d_o$ = \texttt{Obj}[self].$m$($d_i$); \textbf{return} $d_o$; \}
\end{figure}
Here self evaluates to $t$ whenever the method is run by thread $t$.

\floatname{algorithm}{Event}
\renewcommand{\thealgorithm}{}
\setalgorithmicfont{\footnotesize}
\begin{wrapfigure}{L}{0.35\textwidth}
 \begin{minipage}{0.35\textwidth}
 \tiny{
  \begin{algorithm}[H]
   \caption{$d_o=m(d_i)$}
\begin{algorithmic}
\IF {self=$i$}
 \STATE $d_o \gets \mathtt{Obj}$[self].$m$($d_i$);
 \STATE newseq $\gets$ lseq$\cdot m$($d_i$,$d_o$);
 \IF {notRobust(lseq)} 
  \STATE atomic <commit(lseq)>;
  \STATE atomic <$d_o \gets$ O.$m$($d_i$)>;
  \STATE lseq $\gets\varepsilon$;
 \ELSE 
  \STATE lseq $\gets$ newseq; 
 \ENDIF
\ELSE 
 \STATE atomic <$d_o \gets$ O.$m$($d_i$)>;
\ENDIF
\RETURN $d_o$;
\end{algorithmic}
  \end{algorithm}
  }
 \end{minipage}
\caption{The template for $m(d_i,d_o)$ in $\myimplclose$.}
\mylabel{fig:closed}
\end{wrapfigure}
As for singular implementations, we use the template given in Fig.~\ref{fig:closed}.
Intuitively, there is one non-deterministically assigned thread id ($i$) which each thread checks whether is equal to its own.
There is a local object for each thread, like the isolated implementation template explained above.
Additionally, there is another object, \texttt{O}, visible to all threads.
If the thread with identifier $t\neq i$ invokes a method $m$ with input $d_i$, then it applies $m(d_i)$ on \texttt{O} atomically (e.g. performing the operation only after acquiring a global lock and releasing upon completion).
Otherwise, if $t=i$, then the thread checks whether the sequence it has locally performed so far (kept in the thread-local variable lseq) is robust. 
If not, it proceeds like other threads, atomically applying the method.
If the sequence so far has been robust, the result of applying $m(d_i)$ to it is checked again for robustness.
If appending $m(d_i,d_o)$ to lseq leaves it robust, lseq is updated and $d_o$ which is the result of applying $m(d_i)$ to $\mathcal{D}$ after lseq is returned.
Otherwise, the sequence up to now is atomically applied to \texttt{O} and $t$ becomes fully synchronized.




\section{From SC to Linearizability - Forced Synchronization}
\mylabel{sec:forced-sync}
The previous section showed that the definition of sequentially consistency is too weak.
If it were to be taken as is as the correctness criterion, certain {\em broken} implementations, such as the isolated or singular implementations, would be correct.
We also know that the same implementations are not linearizable.
On the one hand the synchronization required for achieving linearizable data structures is also the culprit for non-scalable implementations.
On the other hand complete disregard for synchronization allowed by sequential consistency leaves us with pathological implementations.
In this section we propose a way to quantitatively bridge the gap from sequential consistency to linearizability.

Our idea is to limit the number of consecutive (total) events that a thread can execute before being forced to synchronize.
Let $\thrhist{h}$ be a threaded $\mathcal D$-history.
For any $t\in T$, let $e_j$ be the $i^{th}$ event of $t$ if $\history{h}(t)(i)=e_j$.
For any SC threaded $\mathcal D$-history $\thrhist{h}$, let us call a sequential threaded $\mathcal D$-history $\thrhist{s}$ a {\em serialization} of $\thrhist{h}$ if for all $t\in T$, $\history{s}(t)=\history{h}(t)$.

\begin{definition}[$k$-serial]
Let a threaded $\mathcal D$-history $\thrhist{h}$ be SC. 
Then, $\thrhist{h}$ is $k$-serial if there exists a serialization $\thrhist{s}$ of $\thrhist{h}$ such that for any $t\in T$, $e,e'\in \history{h}$ whenever $e'$ is the $i^{th}$ event of $t$ and $e'\precOrder{\history{h}}e$, then for all $j\leq i-k$ we have $\history{h}(t)(j)\abOrder{\history{s}}e$. 
An implementation is $k$-SC if all its traces are $k$-serial.
\end{definition}
Informally, a threaded history is $k$-serial if a thread cannot continue execution for more than $k$ events without synchronizing with other threads.
In SC proper, since there is no explicit requirement for synchronization, for any $k\in \mathbb{N}$ one can construct a threaded $\mathcal{D}$-history such that it is not $k$-serial as long as $\mathcal{D}$ has at least one robust sequence.
In linearizability, this bound is by definition 0; i.e. $\thrhist{h}$ is $0$-serial iff $\history{h}$ is linearizable. 
We state these results formally.
\begin{lemma}
Let $\mathcal{D}$ contain at least one robust sequence.
\begin{itemize}
\item For any $k\in\mathbb{N}$, there exists a sequence which is $k+1$-serial but not $k$-serial.
\item If $\history{h}$ is linearizable, then it is $0$-serial.
\end{itemize}
\end{lemma}

\setalgorithmicfont{\footnotesize}
\begin{wrapfigure}{L}{0.35\textwidth}
 \begin{minipage}{0.35\textwidth}
 \tiny{
  \begin{algorithm}[H]
   \caption{$d_o=m(d_i)$}
\begin{algorithmic}
\IF {self=$i$}
 \STATE $d_o \gets \mathtt{Obj}$[self].$m$($d_i$);
 \STATE newseq $\gets$ lseq$\cdot m$($d_i$,$d_o$);
 \IF {notRobust(lseq) \fbox{$\vee$ cnt $\geq k$}}
  \STATE atomic <commit(lseq)>;
  \STATE atomic <$d_o \gets$ O.$m$($d_i$)>;
  \STATE lseq $\gets\varepsilon$; \fbox{cnt=0;}
 \ELSE 
  \STATE lseq $\gets$ newseq; \fbox{cnt++;}
 \ENDIF
\ELSE 
 \STATE atomic <$d_o \gets$ O.$m$($d_i$)>;
\ENDIF
\RETURN $d_o$;
\end{algorithmic}
  \end{algorithm}
  }
 \end{minipage}
\caption{The template for $m(d_i,d_o)$ in $k$-SC.}
\mylabel{fig:k-closed}
\end{wrapfigure}
It is straightforward to implement $k$-SC data structures by modifying the singular implementations given in the previous section (modifications shown by the boxed code of Fig.~\ref{fig:k-closed}).
Events are performed locally without synchronization as long as the the sequence so far has been robust and its length is less than $k$ (the additional disjunct {\small \fbox{cnt $\geq k$}}).
Once either of the conditions is violated, the effects of all events seen so far are committed to the shared data structure.
Until then, the local sequence and its length is updated (the increment {\small \fbox{cnt++}}).
Observe that if $k$ is taken to be 0, the additional disjunct will always evaluate to true, forcing synchronization at each call, thereby guaranteeing linearizability.
%




\section{Conclusion}
\mylabel{sec:conclusion}
We have shown that sequential consistency despite its appeal is too weak to be used as an alternative to linearizability in specifying concurrent data structure correctness.
For almost all well-known data structures, sequentially consistent implementations thereof can have undesirable behavior.
For instance, it is possible for a thread in a sequentially consistent queue implementation to observe the queue as empty regardless of what the other threads are doing.

As a first step to bridge the gap between sequentially consistent and linearizable implementations, we also propose a quantitative constraint to capture implementations that lie between the two consistency conditions.
In a $k$-SC implementation, a thread is allowed to proceed without synchronization only for a determined number of consecutive events after which it is required to synchronize. 

One possible future work is the development of concrete data structures that are $k$-SC and investigate the relation between particular values of $k$ and some notion of overall progress.
Another possibility is to check whether a similar strengthening of other consistency conditions either weaker than (memory models of modern processors, such as x86 or ARM) or incomparable to (e.g. quiescent consistency) sequential consistency is useful.


\bibliographystyle{plain}
\bibliography{sc-datastructures}

\begin{thebibliography}{10}

\bibitem{AKY2010}
Yehuda Afek, Guy Korland, and Eitan Yanovsky.
\newblock Quasi-linearizability: Relaxed consistency for improved concurrency.
\newblock In {\em Proceedings of the 14th International Conference on
  Principles of Distributed Systems}, OPODIS'10, pages 395--410.
  Springer-Verlag, 2010.

\bibitem{ABJ+1993}
Mustaque Ahamad, Rida~A. Bazzi, Ranjit John, Prince Kohli, and Gil Neiger.
\newblock The power of processor consistency.
\newblock In {\em Proceedings of the Fifth Annual ACM Symposium on Parallel
  Algorithms and Architectures}, SPAA '93, pages 251--260. ACM, 1993.

\bibitem{ANB+1995}
Mustaque Ahamad, Gil Neiger, JamesE. Burns, Prince Kohli, and PhillipW. Hutto.
\newblock Causal memory: definitions, implementation, and programming.
\newblock {\em Distributed Computing}, 9(1):37--49, 1995.

\bibitem{AKL+2015}
Dan Alistarh, Justin Kopinsky, Jerry Li, and Nir Shavit.
\newblock The spraylist: A scalable relaxed priority queue.
\newblock In {\em Proceedings of the 20th ACM SIGPLAN Symposium on Principles
  and Practice of Parallel Programming}, PPoPP 2015, pages 11--20. ACM, 2015.

\bibitem{BGY+2014}
Sebastian Burckhardt, Alexey Gotsman, Hongseok Yang, and Marek Zawirski.
\newblock Replicated data types: Specification, verification, optimality.
\newblock In {\em Proceedings of the 41st ACM SIGPLAN-SIGACT Symposium on
  Principles of Programming Languages}, POPL '14, pages 271--284. ACM, 2014.

\bibitem{Goo1991}
James~R. Goodman.
\newblock Cache consistency and sequential consistency.
\newblock Technical Report 1006, Computer Sciences Department, University of
  Wisconsin, February 1991.

\bibitem{HHK+2013}
Andreas Haas, Thomas~A. Henzinger, Christoph~M. Kirsch, Michael Lippautz,
  Hannes Payer, Ali Sezgin, and Ana Sokolova.
\newblock Distributed queues in shared memory: Multicore performance and
  scalability through quantitative relaxation.
\newblock In {\em Proceedings of the ACM International Conference on Computing
  Frontiers}, CF '13, pages 17:1--17:9. ACM, 2013.

\bibitem{HS1992}
Abdelsalam Heddaya and Himanshu Sinha.
\newblock Coherence, non-coherence and local consistency in distributed shared
  memory for parallel computing.
\newblock Technical report, Computer Science Department, Boston University,
  1992.

\bibitem{HKP+2013}
Thomas~A. Henzinger, Christoph~M. Kirsch, Hannes Payer, Ali Sezgin, and Ana
  Sokolova.
\newblock Quantitative relaxation of concurrent data structures.
\newblock In {\em Proceedings of the 40th Annual ACM SIGPLAN-SIGACT Symposium
  on Principles of Programming Languages}, POPL '13, pages 317--328. ACM, 2013.

\bibitem{HW1990}
Maurice~P. Herlihy and Jeannette~M. Wing.
\newblock Linearizability: A correctness condition for concurrent objects.
\newblock {\em ACM Trans. Program. Lang. Syst.}, 12(3):463--492, July 1990.

\bibitem{JR2014}
Radha Jagadeesan and James Riely.
\newblock Between linearizability and quiescent consistency.
\newblock In {\em Automata, Languages, and Programming}, Lecture Notes in
  Computer Science, pages 220--231. Springer Berlin Heidelberg, 2014.

\bibitem{KLP2013}
ChristophM. Kirsch, Michael Lippautz, and Hannes Payer.
\newblock Fast and scalable, lock-free k-fifo queues.
\newblock In {\em Parallel Computing Technologies}, Lecture Notes in Computer
  Science, pages 208--223. Springer Berlin Heidelberg, 2013.

\bibitem{Lam1979}
L.~Lamport.
\newblock How to make a multiprocessor computer that correctly executes
  multiprocess programs.
\newblock {\em IEEE Trans. Comput.}, 28(9):690--691, 1979.

\bibitem{LS1988}
R.J. Lipton and J.S. Sandbert.
\newblock Pram: A scalable shared memory.
\newblock Technical Report TR-180-88, Department of Computer Science, Princeton
  University, August 1988.

\bibitem{RSD2014}
Hamza Rihani, Peter Sanders, and Roman Dementiev.
\newblock Multiqueues: Simpler, faster, and better relaxed concurrent priority
  queues.
\newblock {\em CoRR}, abs/1411.1209, 2014.

\bibitem{Sha2011}
Nir Shavit.
\newblock Data structures in the multicore age.
\newblock {\em Commun. ACM}, 54(3):76--84, March 2011.

\end{thebibliography}

\end{document}